\documentclass{llncs}
\usepackage{amssymb, amsmath,graphicx,graphics,caption,subcaption,enumerate,comment}
\usepackage[utf8]{inputenc} 
\usepackage[T1]{fontenc}
\usepackage{tikz}
\usetikzlibrary{calc}
\usepackage{hyperref}
\usepackage[noadjust,nospace]{cite}

\newcommand{\NP}{{\sf NP}}

\newcommand{\vd}{{\sf vd}}
\newcommand{\ed}{{\sf ed}}
\newcommand{\ec}{{\sf ec}}

\newtheorem{open}{Open Problem}

\title{Contracting Bipartite Graphs to Paths and Cycles\thanks{This paper received support from 
EPSRC (EP/K025090/1) and the Leverhulme Trust (RPG-2016-258).
An extended abstract of it was published in the proceedings of EuroComb 2017~\cite{DP17}.}}
\author{Konrad K. Dabrowski\inst{1} \and Dani\"el Paulusma\inst{1}}
\institute{School of Engineering and Computing Sciences, Durham University,\\
Science Laboratories, South Road,\\
Durham DH1 3LE, United Kingdom\\
\texttt{\{konrad.dabrowski,daniel.paulusma\}@durham.ac.uk}
}

\begin{document}
\maketitle

\begin{abstract}
Testing if a given graph~$G$ contains the $k$-vertex path~$P_k$ as a minor or as an induced minor is trivial for every fixed integer $k\geq 1$.
However, the situation changes for the problem of checking if a graph can be modified into~$P_k$ by using only edge contractions.
In this case the problem is known to be \NP-complete even if $k=4$.
This led to an intensive investigation for testing contractibility on restricted graph classes.
We focus on bipartite graphs.
Heggernes, van 't Hof, L{\'{e}}v{\^{e}}que and Paul proved that the problem stays \NP-complete for bipartite graphs if $k=6$.
We strengthen their result from $k=6$ to $k=5$.
We also show that the problem of contracting a bipartite graph to the $6$-vertex cycle $C_6$ is \NP-complete.
The cyclicity of a graph is the length of the longest cycle the graph can be contracted to.
As a consequence of our second result, determining the cyclicity of a bipartite graph is \NP-hard.

\end{abstract}

\section{Introduction}

Algorithmic problems for deciding whether the structure of a graph~$H$ appears as a ``pattern'' within the structure of another graph~$G$ are well studied.
Here, the definition of a pattern depends on the set~$S$ of graph operations that we are allowed to use.
Basic graph operations include vertex deletion~\vd, edge deletion \ed\ and edge contraction \ec.
Contracting an edge~$uv$ means that we delete the vertices~$u$ and~$v$ and introduce a new vertex with neighbourhood $(N(u)\cup N(v))\setminus \{u,v\}$ (note that no multiple edges or self-loops are created in this way).
A graph~$G$ contains a graph~$H$ as a {\it minor} if~$H$ can be obtained from~$G$ using operations from $S=\{\vd,\ed,\ec\}$.
For $S=\{\vd,\ec\}$ we say that~$G$ contains~$H$ as an {\it induced minor}, and for $S=\{\ec\}$ we say that~$G$ contains~$H$ as a {\it contraction}.
For a fixed graph~$H$ (that is, $H$ is not part of the input), the corresponding three decision problems are denoted by $H$-{\sc Minor}, $H$-{\sc Induced Minor} and $H$-{\sc Contractibility}, respectively.

A celebrated result by Robertson and Seymour~\cite{RS95} states that the $H$-{\sc Minor} problem can be solved in cubic time for every fixed pattern graph~$H$.
The problems {\sc $H$-Induced Minor} and {\sc $H$-Contractibility} are harder.
Fellows et al.~\cite{FKMP95} gave an example of a graph~$H$ on 68 vertices for which {\sc $H$-Induced Minor} is \NP-complete, whereas Brouwer and Veldman~\cite{BV87} proved that {\sc $H$-Contractibility} is \NP-complete even when $H=P_4$ or $H=C_4$ (the graphs~$C_k$ and~$P_k$ denote the cycle and path on~$k$ vertices, respectively).
Both complexity classifications are still not settled, as there are many graphs~$H$ for which the complexity is unknown (see also~\cite{HKPST12}).

We observe that $P_k$-{\sc Induced Minor} and $C_k$-{\sc Induced Minor} are polynomial-time solvable for all~$k$; it suffices to check if~$G$ contains~$P_k$ as an induced subgraph, that is, if~$G$ is not {\it $P_k$-free}, or if~$G$ contains an induced cycle of length at least~$k$.
In order to obtain similar results to those for minors and induced minors, we need to restrict the input of the $P_k$-{\sc Contractibility} and $C_k$-{\sc Contractibility} problems to some special graph class.

Of particular relevance is the closely related problem of determining the {\it cyclicity}~\cite{Ha99} of a graph, that is, the length of a longest cycle to which a given graph can be contracted.
Cyclicity was introduced by Blum~\cite{Bl82} under the name {\it co-circularity,} due to a close relationship with a concept in topology called circularity (see also~\cite{BBDG80}).
Later Hammack~\cite{Ha99} coined the current name for the concept and gave both structural results and polynomial-time algorithms for a number of special graph classes.
He also proved that the problem of determining the cyclicity is \NP-hard for general graphs~\cite{Ha02}.

Van 't Hof, Paulusma and Woeginger~\cite{HPW09} proved that the $P_4$-{\sc Contractibility} problem is \NP-complete for $P_6$-free graphs, but polynomial-time solvable for $P_5$-free graphs.
Their results can be extended in a straightforward way to obtain a complexity dichotomy for $P_k$-{\sc Contractibility} restricted to $P_\ell$-free graphs except for one missing case, namely when $k=5$ and $\ell=6$.
Fiala, Kami\'nski and Paulusma~\cite{FKP13} proved that $P_k$-{\sc Contractibility} is \NP-complete on line graphs (and thus for claw-free graphs) for $k\geq 7$ and polynomial-time solvable on claw-free graphs (and thus
for line graphs) for $k\leq 4$.
The problems of determining the computational complexity for the missing cases $k=5$ and $k=6$ were left open.
The same authors also proved that $C_6$-{\sc Contractibility} is \NP-complete for claw-free graphs~\cite{Ha02}, which implies that determining the cyclicity of a claw-free graph is \NP-hard.

Hammack~\cite{Ha99} proved that {\sc $C_k$-Contractibility} is polynomial-time solvable on planar graphs for every $k\geq 3$.
Later, Kami\'nski, Paulusma and Thilikos~\cite{KPT10} proved that {\sc $H$-Contractibility} is polynomial-time solvable on planar graphs for every graph~$H$.
Golovach, Kratsch and Paulusma~\cite{GKP13} proved that the {\sc $H$-Contractibility} problem is polynomial-time solvable on AT-free graphs for every triangle-free graph~$H$.
Hence, as $C_3$-{\sc Contractibility} is readily seen to be polynomial-time solvable for general graphs, {\sc $C_k$-Contractibility} and {\sc $P_k$-Contractibility} are polynomial-time solvable on AT-free graphs for every integer $k\geq 3$.
Heggernes et al.~\cite{HHLP14} proved that $P_k$-{\sc Contractibility} is polynomial-time solvable on chordal graphs for every $k\geq 1$.
Later, Belmonte et al.~\cite{BGHHKP14} proved that {\sc $H$-Contractibility} is polynomial-time solvable on chordal graphs for every graph~$H$.
Heggernes et al.~\cite{HHLP14} also proved that $P_6$-{\sc Contractibility} is \NP-complete even for the class of bipartite graphs.

\subsection*{Research Question} 

We consider the class of bipartite graphs, for which we still have a limited understanding of the {\sc $H$-Contractibility} problem.
In contrast to a number of other graph classes, as discussed above, bipartite graphs are not closed under edge contraction, which means that getting a handle on the $H$-{\sc Contractibility} problem is more difficult.
We therefore focus on the $H=P_k$ and $H=C_k$ cases of the following underlying research question for $H$-{\sc Contractibility} restricted to bipartite graphs:

\medskip
\noindent
{\it Do the computational complexities of $H$-{\sc Contractibility} for general graphs and bipartite graphs coincide for every graph~$H$?}

\medskip
\noindent
This question belongs to a more general framework, where we aim to research whether for graph classes not closed under edge contraction, one is still able to obtain ``tractable'' graphs~$H$, for which the $H$-{\sc Contractibility} problem is \NP-complete in general.
For instance, claw-free graphs are not closed under edge contraction.
However, there does exist a graph~$H$, namely $H=P_4$, such that $H$-{\sc Contractibility} is polynomial-time solvable on claw-free graphs and \NP-complete for general graphs.
Hence, being claw-free at the start is a sufficiently strong property for $P_4$-{\sc Contractibility} to be polynomial-time solvable, even though applying contractions might take us out of the class.
It is not known whether being bipartite at the start is also sufficiently strong.

\subsection*{Our Contribution}
We recall that the $H$-{\sc Contractibility} problem is already \NP-hard if $H=C_4$ or $H=P_4$.
Hence, with respect to our research question we will need to consider small graphs~$H$.
While we do not manage to give a conclusive answer, we do improve upon the aforementioned result from 
Heggernes~et al.~\cite{HHLP14} on bipartite graphs by showing in Section~\ref{s-main} that even $P_5$-{\sc Contractibility} is \NP-complete for bipartite graphs.

\begin{theorem}\label{t-main}
$P_5$-{\sc Contractibility} is \NP-complete for bipartite graphs.
\end{theorem}

We also have the following result, which we prove in Section~\ref{s-main2}.

\begin{theorem}\label{t-main2}
The $C_6$-{\sc Contractibility} problem is \NP-complete for bipartite graphs.
\end{theorem}

We observe that if a graph can be contracted to~$C_k$ for some integer $k\geq 3$, it can also be contracted to~$C_\ell$ for any integer $3\leq \ell\leq k$.
Hence, as an immediate consequence of Theorem~\ref{t-main2}, we obtain the following result.

\begin{corollary}
The problem of determining whether the cyclicity of a bipartite graph is at least~$6$ is \NP-complete.
\end{corollary}

\section{A Known Lemma}\label{s-pre}

A graph~$G$ contains a graph~$H$ as a contraction if and only if for every vertex~$h$ in~$V_H$ there is a nonempty subset $W(h)\subseteq V_G$ of vertices in~$G$ such that:
\begin{itemize} 
\item[$\bullet$] $G[W(h)]$ is connected; 
\item [$\bullet$] the set ${\cal W}=\{W(h)\; |\; h\in V_H\}$ is a partition of~$V_G$; and
\item[$\bullet$] for every $h_i,h_j\in V_H$, there is at least one edge between the witness sets~$W(h_i)$ and~$W(h_j)$ in~$G$ if and only if~$h_i$ and~$h_j$ are adjacent in~$H$.
\end{itemize}
The set~$W(h)$ is an $H$-{\it witness set} of~$G$ for~$h$, and~${\cal W}$ is said to be an {\it $H$-witness structure} of~$G$.
If for every $h\in V_H$ we contract the vertices in~$W(h)$ to a single vertex, then we obtain the graph~$H$.
Witness sets~$W(h)$ may not be uniquely defined, as there could be different sequences of edge contractions that modify~$G$ into~$H$.
A pair of vertices $(u,v)$ of a graph~$G$ is $P_\ell$-\emph{suitable} for some integer $\ell\geq 3$ if and only if~$G$ has a $P_\ell$-witness structure~${\cal W}$ with $W(p_1)=\{u\}$ and $W(p_\ell)=\{v\}$, where $P_\ell=p_1\dots p_\ell$.
See Figure~\ref{f-p4witness} for an example.

\begin{figure}
  \centering
  \includegraphics[scale=1.2]{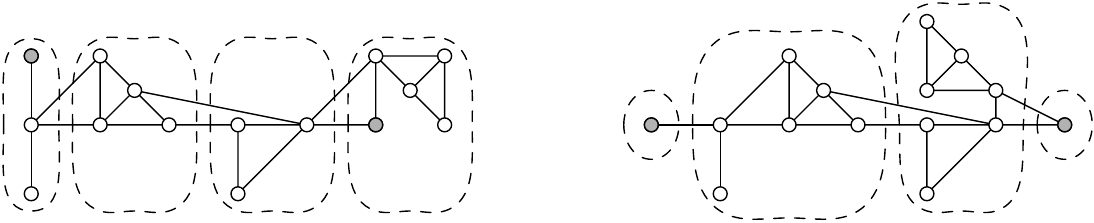}
  \caption{Two $P_4$-witness structures of a graph; the grey vertices form a $P_4$-suitable pair~\cite{HPW09}.}\label{f-p4witness}
\end{figure}

\begin{lemma}[\cite{HPW09}]\label{l-outer}
For $\ell\geq 3$, a graph~$G$ is $P_\ell$-contractible if and only if~$G$ has a $P_\ell$-suitable pair.
\end{lemma}

\section{The Proof of Theorem~\ref{t-main}}\label{s-main}

In this section we prove that $P_5$-{\sc Contractibility} is \NP-complete for bipartite graphs.
The $P_5$-{\sc Contractibility} problem restricted to bipartite graphs is readily seen to be in \NP.
Hence what remains is to prove \NP-hardness.

Let $(Q,{\mathcal S})$ be a hypergraph, where~$Q$ is some set of {\it elements} and~${\mathcal S}$ is a set of {\it hyperedges}, which are subsets of~$Q$.
A {\it $2$-colouring} of $(Q,{\mathcal S})$ is a partition $(Q_1,Q_2)$ of~$Q$ with $Q_1\cap S \ne\emptyset$ and $Q_2 \cap S \ne\emptyset$ for every $S\in {\mathcal S}$.
The corresponding decision problem is called {\sc Hypergraph $2$-Colourability} and is well known to be \NP-complete (see~\cite{GJ79}).
Just as in the proof of~\cite{BV87} for \NP-hardness of {\sc $P_4$-Contractibility} for general graphs, we will reduce from {\sc Hypergraph $2$-Colourability}.
In fact, just as the construction in the proof~\cite{HHLP14} for $P_6$-{\sc Contractibility} for bipartite graphs, our construction borrows elements from~\cite{BV87}, but is more advanced.

Let $(Q,{\mathcal S})$ be a hypergraph with $Q=\{q_1,\ldots,q_m\}$ and ${\mathcal S}=\{S_1,\ldots,S_n\}$.
We may assume without loss of generality that $n\geq 2$, $S_i\neq \emptyset$ for each~$S_i$ and $S_n=Q$.
Given the pair $(Q,{\mathcal S})$, we will construct a graph $G=(V,E)$ in the following way; see Figure~\ref{fig:QS} for an example.
\begin{itemize}
\item
Construct the {\it incidence graph} of $(Q,{\mathcal S})$.
This is a bipartite graph with partition classes~$Q$ and~${\cal S}$, and an edge between two vertices~$q_i$ and~$S_j$ if and only if $q_i\in S_j$.
\item
Add a set ${\mathcal S}'=\{S_1',\ldots,S_n'\}$ of~$n$ new vertices.
Add an edge between~$q_i$ and~$S_j'$ if and only if $q_i\in S_j$.
We say that~$S_j'$ is a {\it copy} of~$S_j$ and say that it represents a hyperedge that contains the same elements as~$S_j$.
\item
Add an edge between every~$S_j$ and~$S_k'$, that is, the subgraph induced by ${\cal S}\cup {\cal S}'$ is complete bipartite.
\item 
Subdivide each edge~$q_iS_j$, that is, remove the edge~$q_iS_j$ and replace it by a new vertex~$q^i_j$ with edges~$q^i_jq_i$ and~$q^i_jS_j$.
Let~$Q'$ consist of all the vertices~$q^i_j$.
\item
Add three new vertices $q^*$, $u_1$ and~$u_2$ and edges $q^*u_1$, $q^*u_2$.
\item Add an edge between~$q^*$ and every~$q^i_j$.
\item Add an edge between~$u_1$ and every~$S_j$, and an edge between~$u_2$ and every~$S_j$.
\item Add two new vertices~$v$ and~$w$.
Add the edges~$u_1v$ and~$u_2v$, and also an edge between~$w$ and every~$S_j'$.
\end{itemize}

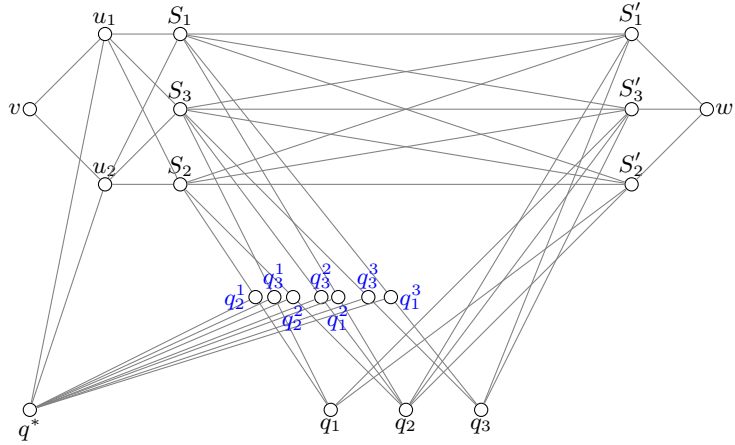
\begin{figure}
\begin{center}
\begin{tikzpicture}
\tikzstyle{vertex}=[circle,draw=black, fill=black, minimum size=5pt, inner sep=1pt]
\tikzstyle{edge} =[draw,-,black]
\foreach \pos/\name / \label / \posn / \dist  in {
{(-1,0)/v/$v$/{left}/0},
{(0,1)/u1/$u_1$/{above}/0},
{(0,-1)/u2/$u_2$/{above}/0},
{(-1,-4)/q/$q^*$/{below}/0},
{(1,1)/s1/$S_1$/{above}/0},
{(1,0)/s2/$S_3$/{above}/0},
{(1,-1)/s3/$S_2$/{above}/0},
{(7,1)/ss1/$S'_1$/{above}/0},
{(7,0)/ss2/$S'_3$/{above}/0},
{(7,-1)/ss3/$S'_2$/{above}/0},
{(8,0)/w/$w$/{right}/0},
{(3,-4)/q1/$q_1$/{below}/0},
{(4,-4)/q2/$q_2$/{below}/0},
{(5,-4)/q3/$q_3$/{below}/0}}
      {\node[vertex, fill=white] (\name) at \pos {};
       \node [\posn=\dist] at (\name) {\label};
       }
\foreach \source/ \dest  in {v/u1, v/u2, q/u1, q/u2, u1/s1, u1/s2, u1/s3, u2/s1, u2/s2, u2/s3, s1/ss1, s1/ss2, s1/ss3, s2/ss1, s2/ss2, s2/ss3, s3/ss1, s3/ss2, s3/ss3, ss1/w, ss2/w, ss3/w, s2/q3, s3/q2, s3/q1, ss2/q3, ss3/q2, ss3/q1, s1/q3, s1/q2, ss1/q3, ss1/q2, s2/q2, s2/q1, ss2/q2, ss2/q1}
       \path[edge, black, thin, color=black!50!white] (\source) --  (\dest);     

\foreach \pos/\name / \label / \posn / \dist  in {
{($(q1)!0.5!(s3)$)/q13/$q^1_3$/{left}/0},
{($(q2)!0.5!(s3)$)/q23/$q^2_3$/{below}/0},
{($(q1)!0.375!(s2)$)/q12/$q^1_2$/{above}/0},
{($(q2)!0.375!(s2)$)/q22/$q^2_2$/{above}/0},
{($(q3)!0.375!(s2)$)/q32/$q^3_2$/{above}/0},
{($(q3)!0.3!(s1)$)/q31/$q^3_1$/{right}/0},
{($(q2)!0.3!(s1)$)/q21/$q^2_1$/{below}/0}}
      { \node[vertex, fill=white] (\name) at \pos {};
       }

\foreach \source/ \dest  in {q13/q, q23/q, q12/q, q22/q, q32/q, q31/q, q21/q}
       \path[edge, black, thin, color=black!50!white] (\source) --  (\dest);     

\foreach \pos/\name / \label / \posn / \dist  in {
{($(q1)!0.5!(s3)$)/q13/$q^1_2$/{left}/0},
{($(q2)!0.5!(s3)$)/q23/$q^2_2$/{below}/0},
{($(q1)!0.375!(s2)$)/q12/$q^1_3$/{above}/0},
{($(q2)!0.375!(s2)$)/q22/$q^2_3$/{above}/0},
{($(q3)!0.375!(s2)$)/q32/$q^3_3$/{above}/0},
{($(q3)!0.3!(s1)$)/q31/$q^3_1$/{right}/0},
{($(q2)!0.3!(s1)$)/q21/$q^2_1$/{below}/0}}
      { 
       \node [\posn=\dist, color=blue] at (\name) {\label};
       }
\end{tikzpicture}
\end{center}
\caption{\label{fig:QS}The graph~$G$ corresponding to the instance of {\sc Hypergraph $2$-Colourability} with $Q=\{q_1,q_2,q_3\}$ and ${\cal S}= \{\{q_2,q_3\},\{q_1,q_2\},\{q_1,q_2,q_3\}\}$.}
\end{figure}

The {\it distance} between two vertices in a graph is the number of edges of a shortest path between them.
The {\it diameter} of a graph is the maximum distance over all pairs of vertices in it.
We note that the graph~$G$ may have arbitrarily large induced paths (alternating between vertices in~$Q$ and~${\cal S}$).
However, as we will check in the proof of Lemma~\ref{l-hard}, $G$ has diameter~$4$, and this property will be crucial.
We first prove the following lemma.

\begin{lemma}\label{l-bipartite}
The graph~$G$ is bipartite.
\end{lemma}

\begin{proof}
We partition~$V$ into $A=\{q^*,v,w\} \cup {\cal S}\cup Q$ and $B=\{u_1,u_2\} \cup {\mathcal S'} \cup Q'$, and note that~$G$ contains no edge between any two vertices in~$A$ and no edge between any two vertices in~$B$.\qed
\end{proof}

\begin{lemma}\label{l-hard}
The hypergraph $(Q,{\mathcal S})$ has a $2$-colouring if and only if the graph~$G$ contains~$P_5$ as a contraction.
\end{lemma}

\begin{proof}
Let $P$ be on a path on five vertices $p_1, \ldots, p_5$ in that order.
First suppose that $(Q,{\mathcal S})$ has a $2$-colouring $(Q_1,Q_2)$.
We define
$W(p_1)=\{v\}$,
$W(p_2)=\{q^*,u_1,u_2\}$,
$W(p_3)={\cal S}\cup Q_1 \cup Q'$,
$W(p_4)= {\cal S}'\cup Q_2$ and
$W(p_5)=\{w\}$.
We note that the sets $W(p_1),\ldots, W(p_5)$ are pairwise disjoint.
Moreover, not only $W(p_1)$, $W(p_2)$ and~$W(p_5)$, but also~$W(p_3)$ and~$W(p_4)$ induce connected subgraphs of~$G$, as~$S_n$ and~$S_n'$ are connected to every vertex in~$Q$ by definition (either via a path of length~2 containing a vertex of $Q'$ or directly via an edge).
We also observe that there are no edges between vertices from~$W(p_1)$ and vertices from $W(p_3)\cup W(p_4)\cup W(p_5)$, no edges between vertices from~$W(p_2)$ and vertices from $W(p_4)\cup W(p_5)$ and no edges between vertices from~$W(p_3)$ and vertices from~$W(p_5)$.
We combine these observations with the existence of edges (for instance, $vu_1$, $u_1S_1$, $S_1S_1'$ and~$S_1'w$) between the two consecutive sets~$W(p_i)$ and~$W(p_{i+1})$ for $i=1,\ldots,4$ to conclude that the sets $W(p_1),\ldots,W(p_5)$ form a $P_5$-witness structure of~$G$.

Now suppose that~$G$ contains~$P_5$ as a contraction.
Then, by Lemma~\ref{l-outer}, we find that~$G$ has a $P_5$-witness structure~${\cal W}$, where $W(p_1)=\{x\}$ and $W(p_5)=\{y\}$ for some vertices~$x$ and~$y$.
We refer to Table~\ref{t-distance} for the distances between vertices of different types.
In this table, entries for a vertex and a set or for two sets display the maximum possible distance between them.
For instance, the entry for~${\cal S}$ and~$Q$ is~$2$, as the maximum distance between a vertex in~${\cal S}$ and a vertex in~$Q$ is~$2$.
We also note, for instance, that the distance between any two vertices in~$Q$ is~$2$, because~$S_n'$ is adjacent to every vertex of~$Q$.

From Table~\ref{t-distance} we can see that there are three possible choices for the pair $\{x,y\}$, which must be of distance at least~$4$ from each other in~$G$, 
namely $\{x,y\}=\{v,q_i\}$ for any~$q_i$, $\{x,y\}=\{q^*,w\}$ or $\{x,y\}=\{v,w\}$.
We discuss each of these cases below.

\begin{table}[h]
\centering
\begin{tabular}{|c|c|c|c|c|c|c|c|c|c|}
\hline
                      &$u_1$ &$u_2$ &$v$ &$w$ &${\cal S}$ &${\cal S}'$ &$Q$ &$Q'$ &$q^*$\\
\hline
       $u_1$      &    0      &   2        &  1      &   3    &      1           &   2              &   3    &     2    &      1   \\
       \hline
       $u_2$     &          &   0        &   1     &  3     &    1             &     2            &    3   &   2      &     1    \\
            \hline
          $v$     &           &           & 0       &    {\bf 4}   &    2             &     3            &  {\bf 4}     &    3     &  2       \\
               \hline
         $w$     &           &           &        &  0     &      2           &   1              &  2     &     3    &  {\bf 4}       \\ 
              \hline
${\cal S}$     &           &           &        &       &       2          &      1           &   2    &   3      &   2      \\
     \hline
${\cal S}'$     &          &           &        &       &                 &      2          &   3    &    2     &  3       \\
     \hline
       $Q$       &          &           &        &       &                 &                 &  2    &    3     &  2       \\
            \hline
       $Q'$      &          &            &        &       &                 &                 &       &    2    &   1      \\
            \hline
       $q^*$    &          &            &        &       &                 &                 &       &         &  0        \\
            \hline
            \end{tabular}\\[3pt]
\caption{The (maximum) distances between two different (types of) vertices in~$G$.}\label{t-distance}
\end{table}

\noindent
{\bf Case~1.} $\{x,y\}=\{v,q_i\}$.\\
Let $x=v$ and $y=q_i$.
From Table~\ref{t-distance} we find that $\{u_1,u_2\}\subseteq W(p_2)$.
Moreover, ${\cal S}\subseteq W(p_3)$, as every vertex in~${\cal S}$ is of distance~$2$ from both~$v$ and~$q_i$.
As~$W(p_2)$ induces a connected subgraph by definition and~$u_1$ is not adjacent to~$u_2$, this means that~$q^*$ must be in~$W(p_2)$.
However, this is not possible as~$q^*$ is of distance~$2$ from~$q_i$, which is in~$W(p_5)$.
Hence Case~1 is not possible.

\medskip
\noindent
{\bf Case~2.} $\{x,y\}=\{q^*,w\}$.\\
Let $x=q^*$ and $y=w$.
From Table~\ref{t-distance} we find that $Q'\cup \{u_1,u_2\} \subseteq W(p_2)$ and that ${\cal S'}\subseteq W(p_4)$.
The latter, combined with the fact that every vertex of~${\cal S}$ is adjacent to every vertex of~${\cal S}'$, implies that ${\cal S}\cap W(p_2)=\emptyset$.
Any path from a vertex in~$Q'$ to a vertex in $\{u_1,u_2\}$ must contain at least one vertex of ${\cal S}\cup \{q^*\}$.
As $({\cal S}\cup \{q^*\})\cap W(p_2)=\emptyset$, this means that~$W(p_2)$ does not induce a connected subgraph.
This violates the definition of a witness structure, so Case~2 is not possible either.

\medskip
\noindent
{\bf Case~3.} $\{x,y\}=\{v,w\}$.\\
Let $x=v$ and $y=w$.
From Table~\ref{t-distance} we find that $\{u_1,u_2\}\subseteq W(p_2)$.
Moreover, ${\cal S}\subseteq W(p_3)$, as every vertex in~${\cal S}$ is of distance~$2$ from both~$v$ and~$w$.
As~$W(p_2)$ must induce a connected subgraph of~$G$ by definition, this means that $q^*\in W(p_2)$.
From Table~\ref{t-distance} we also find that ${\cal S}'\subseteq W(p_4)$.

By definition, $W(p_3)$ must induce a connected subgraph.
Recall that~${\cal S}$ is an independent set.
Hence, for each~$S_j$, we find that~$W(p_3)$ contains at least one vertex not in~${\cal S}$ that connects~$S_j$ to the other vertices of~${\cal S}$ (recall that by assumption, $n\geq 2$, so there is at least one other vertex in~${\cal S}$ not equal to~$S_j$).
As $\{u_1,u_2\}\subseteq W(p_2)$ and ${\cal S}'\subseteq W(p_4)$, such a vertex can only be in~$Q'$ and we denote it by~$q'(S_j)$.
As every vertex in~$Q'$ has only three neighbours and one of them is~$q^*$, which is in~$W(p_2)$, we find that~$q'(S_j)$ must be adjacent to a vertex $q(S_j)\in Q\cap W(p_2)$ in order to connect~$S_j$ to the other vertices of~${\cal S}$.
Note that $q(S_j)=q(S_k)$ is possible for two vertices~$S_j$ and~$S_k$ with $k\neq j$.

The set~$W(p_3)$ also induces a connected subgraph and~${\cal S}'$ is an independent set of size at least~2.
Hence, for each~$S_j'$, we find that~$W(p_4)$ contains at least one vertex not in~${\cal S}'$ that connects~$S_j'$ to the other vertices of~${\cal S}'$.
As ${\cal S}\subseteq W(p_3)$ and $w\in W(p_5)$, such a vertex can only be in~$Q$ and we denote it by~$q(S_j')$.
Note that $q(S_j')=q(S_k')$ is possible for two vertices~$S_j'$ and~$S_k'$ with $k\neq j$.

Let~$Q_1$ be the subset of~$Q$ that contains all vertices~$q(S_j)$, so~$Q_1$ is contained in~$W(p_3)$.
Similarly, let~$Q_2$ be the subset of~$Q$ that contains all vertices~$q(S_j')$, so~$Q_2$ is contained in~$W(p_4)$.
Each hyperedge~$S_j$ contains~$q(S_j)$ due to the edges $S_jq'(S_j)$ and $q'(S_j)q(S_j)$.
Moreover, each hyperedge~$S_j$ contains~$q(S_j')$ due to the edge $S_j'q(S_j')$ and because~$S_j'$ is a copy of~$S_j$.
Hence~$S_j$ contains both an element from~$Q_1$ and an element from~$Q_2$.
Moreover, $Q_1$ and~$Q_2$ are disjoint.
Hence, $(Q_1,Q_2)$ is a $2$-colouring of $(Q,{\cal S})$ (note that there may be elements of~$Q$ not in $Q_1\cup Q_2$; we can add such elements to either~$Q_1$ or~$Q_2$ in an arbitrary way).
This completes the proof of Lemma~\ref{l-hard}.\qed
\end{proof}

Combining Lemmas~\ref{l-bipartite} and~\ref{l-hard} with the aforementioned observation on membership in \NP\ implies Theorem~\ref{t-main}.

\medskip
\noindent 
{\bf Theorem~\ref{t-main} (restated).} 
{\it $P_5$-{\sc Contractibility} is \NP-complete for bipartite graphs.}

\section{The Proof of Theorem~\ref{t-main2}}\label{s-main2}

In this section we prove Theorem~\ref{t-main2}.
We do this as follows.
Consider the graph~$G$ constructed in Section~\ref{s-main} for a given instance $(Q,{\cal S})$ of {\sc Hypergraph $2$-Colouring}. We remove the vertices~$q^*$ and~$u_2$, and instead add a new vertex~$x$ that we make adjacent to both~$v$ and~$w$.
This yields the graph $G'=(V',E')$.

Removing the edge~$vx$ from $G'$ results in the graph $G'-vx$, which is used in the hardness construction of Heggernes et al.~\cite{HHLP14}) for proving that 
$P_6$-{\sc Contractibility} is \NP-complete. 

\begin{lemma}[\cite{HHLP14}]\label{l-p}
The hypergraph $(Q,{\mathcal S})$ has a $2$-colouring if and only if $G'-vx$ contains~$P_6$ as a contraction.
\end{lemma}

We continue by proving two lemmas for~$G'$ that are similar to the two lemmas of Section~\ref{s-main}.

\begin{lemma}\label{l-bipartite2}
The graph~$G'$ is bipartite.
\end{lemma}

\begin{proof}
We partition~$V'$ into $A'=\{v,w\} \cup {\cal S}\cup Q$ and $B'=\{u_1,x\} \cup {\mathcal S'} \cup Q'$, and note that~$G'$ contains no edge between any two vertices in~$A'$ and no edge between any two vertices in~$B'$.\qed
\end{proof}

\begin{lemma}\label{l-hard2}
The hypergraph $(Q,{\mathcal S})$ has a $2$-colouring if and only if the graph~$G'$ contains~$C_6$ as a contraction.
\end{lemma}

\begin{proof}
Let $C$ be a cycle on six vertices $c_1,\ldots, c_6$ in that order.
First suppose that $(Q,{\mathcal S})$ has a $2$-colouring $(Q_1,Q_2)$.
We define the following witness sets:
$W(c_1)=\{v\}$,
$W(c_2)=\{u_1\}$,
$W(c_3)={\cal S}\cup Q_1 \cup\nobreak Q'$,
$W(c_4)= {\cal S}'\cup Q_2$,
$W(c_5)=\{w\}$ and
$W(c_6)=\{x\}$.
The sets $W(c_1),\ldots,W(c_6)$ are readily seen to form a $C_6$-witness structure of~$G'$.

Now suppose that~$G'$ contains~$C_6$ as a contraction.
The only vertex of distance at least~$3$ from~$S_n'$ in~$G'$ is~$v$ (in particular recall that~$S_n'$ is adjacent to every vertex of~$Q$).
Hence we may assume without loss of generality that $W(c_1)=\{v\}$ and $S_n'\in W(c_4)$.
Then, as the only two neighbours of~$v$ are~$u_1$ and~$x$, we may also assume without loss of generality that $u_1\in W(c_2)$ and $x\in W(c_6)$.
Since~$v$ and~$w$ are the only two neighbours of~$x$, and~$w$ is a neighbour of $S_n'\in W(c_4)$, this means that $w\in W(c_5)$ and thus $W(c_6)=\{x\}$.
The fact that $W(c_1)=\{v\}$ and $W(c_6)=\{x\}$ implies that
$G'-vx$ contains~$P_6$ as a contraction and we may apply Lemma~\ref{l-p}.
\qed
\end{proof}

Combining Lemmas~\ref{l-bipartite2} and~\ref{l-hard2} with the observation that $C_6$-{\sc Contractibility} belongs to \NP\ implies Theorem~\ref{t-main2}.

\medskip
\noindent
{\bf Theorem~\ref{t-main2} (restated).} {\it The $C_6$-{\sc Contractibility} problem is \NP-complete for bipartite graphs.}

\section{Future Work}

We have proved that the $P_5$-{\sc Contractibility} problem is \NP-complete for the class of bipartite graphs, which strengthens a result in~\cite{HHLP14}, where \NP-completeness was shown for $P_6$-{\sc Contractibility} restricted to bipartite graphs.
As $P_3$-{\sc Contractibility} is readily seen to be polynomial-time solvable for general graphs, this leaves us with one stubborn open case, namely $P_4$-{\sc Contractibility}. 

\begin{open}\label{p-1}
Determine the complexity of $P_4$-{\sc Contractibility} for bipartite graphs.
\end{open}

One approach for settling Open Problem~\ref{p-1} would be to first consider {\it chordal bipartite} graphs, which are bipartite graphs in which every induced cycle has length~$4$.
We believe that this is an interesting question on its own.

\begin{open}
\begin{sloppypar}Determine the complexity of $P_4$-{\sc Contractibility} for chordal bipartite graphs.\end{sloppypar}
\end{open}
We also proved that the $C_6$-{\sc Contractibility} problem is \NP-complete for bipartite graphs, which implied that determining the cyclicity of a bipartite graph is \NP-hard.
As mentioned, $C_3$-{\sc Contractibility} is polynomial-time solvable for general graphs.
This leaves us with the following two open cases.

\begin{open}
Determine the complexity of $C_k$-{\sc Contractibility} for bipartite graphs when $4\leq k\leq 5$.
\end{open}
The $2$-{\sc Disjoint Connected Subgraphs} problem takes as input a graph~$G$ and two disjoint subsets~$Z_1$ and~$Z_2$ of~$V(G)$.
It asks whether~$V(G)$ can be partitioned into sets~$A_1$ and~$A_2$, such that $Z_1\subseteq A_1$, $Z_2\subseteq A_2$ and both~$A_1$ and~$A_2$ induce connected subgraphs of~$G$.
Telle and Villanger~\cite{TV13} gave an $O^*(1.7804^n)$-time algorithm for solving this problem, which is known to be \NP-complete even if $|Z_1|=2$~\cite{HPW09}.
Here, the~$O^*$ notation suppresses factors of polynomial order.

By using the algorithm of~\cite{TV13} as a subroutine and Lemma~\ref{l-outer} we immediately obtain an $O^*(1.7804^n)$-time algorithm for solving $P_4$-{\sc Contractibility} on general $n$-vertex graphs.
That is, we guess two non-adjacent vertices~$u$ and~$v$ with non-intersecting neighbourhoods~$N(u)$ and~$N(v)$ to be candidates for a $P_4$-suitable pair and then solve the $2$-{\sc Disjoint Connected Subgraphs} problem for the graph $G- \{u,v\}$ with $Z_1=N(u)$ and $Z_2=N(v)$ (note that we need to consider at most $\binom{n}{2}$ choices of pairs $u,v$).

\begin{proposition}
$P_4$-{\sc Contractibility} can be solved in $O^*(1.7804^n)$ time for (general) graphs on~$n$ vertices.
\end{proposition}

The proof of the aforementioned \NP-completeness result for $2$-{\sc Disjoint Connected Subgraphs} in~\cite{HPW09} can be modified to hold for bipartite graphs by subdividing each edge in the hardness construction.
This brings us to our final open problem.

\begin{open}
Does there exist an exact algorithm for $P_4$-{\sc Contractibility} for bipartite graphs on~$n$ vertices that is faster than $O^*(1.7804^n)$ time?
\end{open}

\end{document}